\documentclass[a4paper]{article}
\usepackage{graphicx}
\usepackage{caption}

\usepackage[numbers]{natbib}
\usepackage{hyperref}

\usepackage{algorithm}
\usepackage{algorithmic}
\usepackage{tikz}
\usetikzlibrary{arrows.meta, positioning}

\usepackage{microtype}

\usepackage{newfloat}
\usepackage{listings}
\DeclareCaptionStyle{ruled}{labelfont=normalfont,labelsep=colon,strut=off} %
\floatstyle{ruled}
\newfloat{listing}{tb}{lst}{}
\floatname{listing}{Listing}

\title{Discounted Cuts: A Stackelberg Approach to Network Disruption}
\author{
  Pål Grønås Drange
  \and
  Fedor V. Fomin
  \and
  Petr Golovach
  \and
  Danil Sagunov\thanks{This author was supported by the Ministry of Economic Development of the Russian Federation (IGK 000000C313925
P4C0002), agreement No.\ 139-15-2025-010.}
}

\usepackage{amsmath,amssymb,amsthm}
\usepackage{cleveref}
\usepackage{xspace}

\usepackage{thmtools}

\usepackage{longtable,booktabs,array}

\newtheorem{lemma}{Lemma}
\newtheorem{proposition}{Proposition}

\newtheorem{observation}{Observation}
\newtheorem{claim}{Claim}

\usepackage{boxedminipage}
\usepackage{framed}
\usepackage{xifthen}
\usepackage{tabularx}
\usepackage{tikz}
\usetikzlibrary{calc,math,patterns,shapes,backgrounds}
\tikzstyle{path} = [color=black,opacity=.30,line cap=round, line join=round, line width=10pt]

\DeclareMathOperator{\dist}{dist}
\newcommand{\PQ}{\ensuremath{\mathsf{Q}}}

\newcommand{\mincheap}{\textsc{MinCut with $k$ Free Cheap Edges}\xspace}
\newcommand{\minstexp}{\textsc{Min $s$-$t$-Cut with $k$ Free Expensive Edges}\xspace}
\newcommand{\minstcheap}{\textsc{Min $s$-$t$-Cut with $k$ Free Cheap Edges}\xspace}
\newcommand{\minexp}{\textsc{MinCut with $k$ Free Expensive Edges}\xspace}

\newcommand{\minstcheapshort}{\textsc{Min $s$-$t$-Cut$-k$ cheap}\xspace}
\newcommand{\mincheapshort}{\textsc{MinCut$-k$ cheap}\xspace}

\newcommand{\minstexpshort}{\textsc{Min $s$-$t$-Cut$-k$ exp}\xspace}

\newcommand{\minexpshort}{\textsc{MinCut$-k$ exp}\xspace}

\newcommand{\maxcheapshort}{\textsc{MaxCut$-k$ cheap}\xspace}
\newcommand{\maxexpshort}{\textsc{MaxCut$-k$ exp}\xspace}

\newcommand{\budg}{\beta}

\newcommand{\minstcut}{\textsc{Min $s$-$t$-Cut}\xspace}
\newcommand{\mincut}{\textsc{MinCut}\xspace}
\newcommand{\maxcut}{\textsc{MaxCut}\xspace}

\newcommand{\costc}{{\sf Cost}^{\rm cheap}}
\newcommand{\coste}{{\sf Cost}^{\rm exp}}

\newcommand{\opt}{{\sf Opt}}

\newcommand{\NP}{\textsf{NP}}
\newcommand{\W}{\textsf{W}}

\newcommand{\Oh}{\mathcal{O}}

\newlength{\RoundedBoxWidth}
\newsavebox{\GrayRoundedBox}
\newenvironment{GrayBox}[1]{\setlength{\RoundedBoxWidth}{.43\textwidth}
    \def\boxheading{#1}
    \begin{lrbox}{\GrayRoundedBox}
       \begin{minipage}{\RoundedBoxWidth}}{   \end{minipage}
    \end{lrbox}
    \begin{center}
    \begin{tikzpicture}\node(Text)[draw=black!20,fill=white,rounded corners,inner sep=2ex,text width=\RoundedBoxWidth]{\usebox{\GrayRoundedBox}};
        \coordinate(x) at (current bounding box.north west);
        \node [draw=white,rectangle,inner sep=3pt,anchor=north west,fill=white]
        at ($(x)+(6pt,.75em)$) {\boxheading};
    \end{tikzpicture}
    \end{center}}

\newenvironment{defproblemx}[2][]{\noindent\ignorespaces \FrameSep=5pt\parindent=0pt\vspace*{-1.5em}
                \ifthenelse{\isempty{#1}}{\begin{GrayBox}{\textsc{#2}}}{\begin{GrayBox}{\textsc{#2}  parameterized by~{#1}}}
                \begin{tabular*}{\textwidth}{@{\hspace{.1em}} >{\itshape} p{1.8em} p{0.85\textwidth} @{}}}{
                \end{tabular*}\end{GrayBox}\ignorespacesafterend
            }

\definecolor{mdred}{rgb}{0.55, 0.0, 0.0}
\definecolor{mdgreen}{rgb}{0.0, 0.42, 0.24}
\definecolor{mdblue}{rgb}{0.02, 0.05, 0.6}
\definecolor{mdyellow}{rgb}{0.8, 0.6, 0.2}

\begin{document}

\maketitle

\begin{abstract}
We study a Stackelberg variant of the classical Most Vital Links problem, modeled as a one-round adversarial game between an attacker and a defender. The attacker strategically removes up to $k$ edges from a flow network to maximally disrupt flow between a source $s$ and a sink $t$, after which the defender optimally reroutes the remaining flow.
To capture this attacker--defender interaction, we introduce a new mathematical model of
\emph{discounted cuts}, in which the cost of a cut is evaluated by excluding its $k$ most expensive edges. This model generalizes the Most Vital Links problem and uncovers novel algorithmic and complexity-theoretic properties.

We develop a unified algorithmic framework for analyzing various forms of discounted cut problems, including minimizing or maximizing the cost of a cut under discount mechanisms that exclude either the $k$ most expensive or the $k$ cheapest edges. While most variants are \NP-complete on general graphs, our main result establishes polynomial-time solvability for all discounted cut problems in our framework when the input is restricted to bounded-genus graphs, a relevant class that includes many real-world networks such as transportation and infrastructure networks.
With this work, we aim to open collaborative bridges between artificial intelligence, algorithmic game theory, and operations research.
\end{abstract}

\section{Introduction}
\label{sec:intro}

Consider the following scenario of securing a network: An
attacker attempts to escalate privileges from server~$s$ to a critical
database server~$t$.  The security team has a limited budget~$\beta$ to
deploy firewalls or physically disconnect cables.  Some links are very
expensive to disable---typically those that are main trunk lines
carrying heavy traffic.  However, the company's cybersecurity insurance
allows the team to discount the cost of disabling up to~$k$ links,
either by covering the most expensive ones or, alternatively, the least
expensive.

As another example, consider a multi-agent system where agents are connected via  conflict relationships. We aim to divide the agents into two teams or coalitions in a way that maximizes inter-group tension. When agents are represented as nodes and interaction strengths as weighted edges, this leads to a natural \textsc{MaxCut} formulation: the goal is to partition the graph to maximize the total weight of edges crossing the cut.
However, in many realistic scenarios, some interactions may be negligible or unreliable---such as low-trust links, wrong information, or edges that can be manipulated by a bounded-perception attacker. In such cases, we would be more interested in solving the \emph{discounted maximum cut}, where the objective function ignores the $k$ weakest (i.e., lowest-weight) edges in the cut and thus to focus on strong, strategically meaningful interactions, while discounting minor or noisy ones.

This gives rise to a Stackelberg variant of the classic
cut problems, where a leader chooses a cut and a
follower modifies the cost of up to~$k$ edges post hoc.  Depending on
whether the follower discounts the most or the least expensive edges, the
effective cost of the cut---and thus the optimal strategy---changes
significantly.  See \Cref{fig:flow-network-plain-weights} for an example.

Although the discount variants of cut problems appear closely related to the standard cut
problems and to classical interdiction models such as the \textsc{Most
  Vital Links} problem (introduced by \citet{wollmer1963methodsdetermining,wollmer1964removingarcs}), they
give rise to novel algorithmic questions.
In particular, the cost function becomes non-additive and adversarially dependent on the edge weights, breaking standard structural properties of cuts.

\begin{figure}[h!]
    \centering
    \begin{tikzpicture}[
        vertex/.style={circle, draw, minimum size=0.5cm},
        font=\small
    ]
        \node[vertex] (a) at (2,1) {$a$};
        \node[vertex] (c) at (2,-1) {$c$};
        \node[vertex] (s) at (0,0) {$s$};

        \node[vertex] (b) at (4,1) {$b$};
        \node[vertex] (d) at (4,-1) {$d$};
        \node[vertex] (t) at (6,0) {$t$};

        \draw[mdblue, very thick] (s) -- (a) node[midway, above left=2pt, black] {\scriptsize 3};
        \draw[mdred, very thick] (a) -- (b) node[midway, above=2pt, black] {\scriptsize 1};
        \draw (b) -- (t) node[midway, above right=2pt] {\scriptsize 4};

        \draw (a) -- (c) node[midway, right=2pt] {\scriptsize 2};

        \draw[mdblue, very thick] (s) -- (c) node[midway, below left=2pt, black] {\scriptsize 3};
        \draw[mdred, very thick] (c) -- (d) node[midway, below=2pt, black] {\scriptsize 5};
        \draw (d) -- (t) node[midway, below right=2pt] {\scriptsize 3};

        \draw (b) -- (d) node[midway, right=2pt] {\scriptsize 3};

        \draw[mdred, very thick] (a) -- (d) node[midway, right=2pt, black] {\scriptsize 1}; 
    \end{tikzpicture}
    \caption{A minimum $s$-$t$-cut in this example is $\{sa, sc\}$ (blue) of cost $6$. For $k=1$, the minimum discounted $s$-$t$-cut is $\{ab, ad, cd\}$ (red), with discounted cost~$2$.
    This corresponds to the edge $cd$ being the \emph{most vital link}, i.e., the removal of~$cd$ leads to the largest decrease in maximum flow.
    }
    \label{fig:flow-network-plain-weights}
\end{figure}
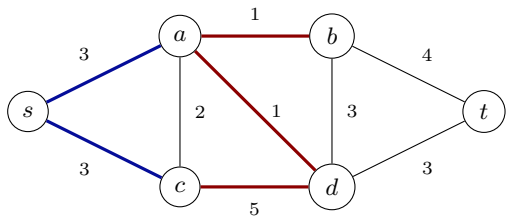

\subsection{Cuts with $k$ Free Edges}

In this paper, we introduce and study several variants of the minimum and maximum cut problems in which the cost of some edges in the cut is discounted.
Specifically, we consider two discount mechanisms: one in which we do not pay for the $k$ most expensive edges, and another in which we do not pay for the $k$ cheapest edges.
These correspond to the two natural game strategies of always discounting the most expensive, and most cheap edges.
We present these discounted-cost cut problems as formal models of network vulnerability and adversarial flow inhibition, providing a suite of complexity and algorithmic results with a particular focus on tractable cases in graphs of bounded genus.
Our findings advance the algorithmic understanding of partial cuts under adversarial or cooperative discounting, opening new directions in network design and interdiction.

Let $S$ be an ordered list of numbers, ordered from smallest to largest.  We define
two discounted cost functions
\[
\coste_k(S)= \sum_i^k S_i \quad \text{and} \quad \costc_k(S)= \sum_{n-k}^n S_i,
\]

\noindent
i.e., $\coste_k$ (resp.\ $\costc_k$) is the sum of the elements in $S$ \emph{except} the $k$ most (resp.\ least) expensive elements.  In this work we will be talking mostly about cuts, so we define the following.
Given a cut $(A,B)$ of a graph $G$ with an (edge) cost function
$c\colon E(G)\rightarrow \mathbb{Z}_{\geq 0}$ and an integer $k\geq 0$,
we define the \emph{discounted cost} of the cut $(A,B)$
as
\begin{align*}
\coste_k(A,B)&=c(E(A,B))-c(R)\\ &=\sum_{e\in E(A,B)} c(e)-\sum_{e\in R}
  c(e),
  \end{align*}
  where $R$ is the set of $k$ most expensive edges of the edge
cut set $E(A,B)$.  (We assume $\coste_k(A,B)=0$ whenever $|E(A,B)|\leq k$.)
We consider the following variant of the classic \textsc{MinCut}
problem
\minstexp:
\emph{Does there exist an $s$-$t$-cut $(A,B)$ of $G$ of discounted
  cost $\coste_k(A,B)\leq \beta$, given  two terminal vertices
  $s$ and $t$, an integer $k\geq 0$, and the budget $\beta\geq 0$.}
Let us note that for $k=0$, \minstexpshort is the problem of finding a
minimum $s$-$t$-cut.

A close well-studied relative of the minimum $s$-$t$-cut problem is the
\textsc{Global Min-cut}, or the edge connectivity of a graph.  Here the
task is to identify the set of edges of minimum weight in a connected
graph whose removal disconnects the graph.  We define the discounted
version of \textsc{Global Min-cut} as follows, \minexp:
Does there exist a cut $(A,B)$ of $G$ such that $\coste_k(A,B)\leq \budg$, given
an integer $k\geq 0$,  and the budget $\budg\geq 0$.
The problem \mincheap, or \mincheapshort for short, is defined similarly.

\bigskip

Another important cut problem is the \textsc{Maximum Cut} problem, where the goal is to partition a graph to maximize the total weight of edges crossing the cut. This problem is significantly more challenging than the corresponding minimization variants. Analogously, we define four versions of the \textsc{Maximum Cut} problem based on the presence or absence of terminals $s$ and $t$, and the choice of discount function (excluding either cheap or expensive edges).
In total, we thus obtain eight distinct discounted cut problems (see \Cref{table:results}).

\subsection{Our Results and Methods}
In this section, we provide an overview of our main results.  A central
theme in our work is the use of a series of reductions that transform
various discounted cut problems into their classic counterparts,
enabling the application of classical techniques in new, cost-aware
settings.  Building on this, we provide complexity classification for
all variants of the discounted cut problem.  Together, these results
offer a comprehensive view of the computational landscape for discounted
cuts and reveal several surprising gaps in complexity between closely
related problems.  We summarize our results in \Cref{table:results}.

\begin{table*}[h]
  \centering
  \scriptsize
    \renewcommand{\arraystretch}{1.3} \setlength{\tabcolsep}{8pt}

    \begin{tabular}{l l l l l}
& \multicolumn{2}{c}{$s$-$t$-cut} & \multicolumn{2}{c}{Global cut} \\
        \cmidrule(r){2-3} \cmidrule(l){4-5}
        & General & Planar & General & Planar \\
        \midrule
        \textsc{Min Exp.}   & W[1]-hard (Thm~\ref{thm:w-one-hard}) & P (Thm~\ref{thm:planar-vital}) & P (Thm~\ref{thm:global_mincut_p}) & P (Thm~\ref{thm:global_mincut_p}) \\
        \textsc{Max Cheap}  & paraNP-hard & P \textit{(i)} (Thm~\ref{thm:global_genus}) & paraNP-hard & P \textit{(i)} (Thm~\ref{thm:global_genus}) \\
        \textsc{Min Cheap}  & P (\Cref{thm:gen-alg}) & P (Thm~\ref{thm:gen-alg}) & P (Thm~\ref{thm:gen-alg}) & P (Thm~\ref{thm:gen-alg}) \\
        \textsc{Max Exp.}   & paraNP-hard & P \textit{(i)} (Thm~\ref{thm:global_genus}) & paraNP-hard & P (Thm~\ref{thm:gen-alg}) \\
        \bottomrule
    \end{tabular}
    \caption{Complexity classification of eight discounted-cut problems (minimum/maximum cut, $s$-$t$/global cut, $k$~cheap/$k$~expensive discounting), each studied on general and planar (bounded-genus) graphs, resulting in 16 complexity entries. The complexity is parameterized by $k$; entries marked ``$(i)$'' indicate tractability results known only for polynomial costs.}
    \label{table:results}
\end{table*}

Surprisingly, despite the long history of \minstexpshort\footnote{\minstexpshort{} is also known as \textsc{Most Vital Links}.}, we have not
found an \textsf{NP}-hardness proof in the literature.  The only known
complexity lower bound is an \textsf{NP}-hardness result for a more
general \emph{network interdiction} model, in which each edge $e$ has a
resource cost $r(e)$.  The goal in this model is to reduce the capacity
of every $s$-$t$-cut by deleting a set of edges whose total resource cost
does not exceed a threshold $R$.  In this setting, \minstexpshort is a
special case where $r(e) = 1$ for every edge of $G$ and $R = k$.  A
simple reduction from \textsc{Knapsack} shows that this more general
interdiction problem is \emph{weakly}
\textsf{NP}-complete~\cite{Ball1989vital,wood1993deterministicnetwork}.
Our first result significantly strengthens the known intractability
bounds by establishing \textsf{NP}-completeness of \minstexpshort when
edge costs are restricted to one and two.

\begin{restatable}{theorem}{thmhardness}
 \label{thm:w-one-hard}
 \minstexpshort is \textsf{NP}-complete for instances where each edge
 has cost one or two.  Moreover, the problem is \textsf{W[1]}-hard
 parameterized by $k$ for instances with integer edge costs upper
 bounded by a polynomial of the size of the input graph.
\end{restatable}

While \minstexpshort{} is \textsf{NP}-complete on general graphs, our next
theorem establishes that on graphs of bounded genus, \emph{all} eight
variants of the discounted cut problem are solvable in polynomial time
when the maximum edge cost is polynomial in $n$.  The proof relies on a
generic approach that reduces discounted cuts with polynomial weights to
their classic counterparts on bounded-genus graphs.  This reduction
builds on an algebraic technique
\cite{galluccio2001optimization} for handling generating functions of
cuts in bounded-genus graphs.

\begin{restatable}{theorem}{thmglobalgenus}\label{thm:global_genus}
  Each of the eight problems
  \begin{itemize}
  \item \textsc{Min/MaxCut-$k$ cheap/exp},
  \item \textsc{Min/Max $s$-$t$-Cut-$k$ cheap/exp}
  \end{itemize}
  restricted to graphs of genus $g$, given with the corresponding
  embedding, admits a randomized algorithm of running time
  $$ \left(4^{g}\cdot n^{1.5} +m\cdot M\right)\cdot m^2\cdot
  M\cdot\log^{\Oh(1)}(n+M),$$ where $M$ is the total edge cost.  The
  algorithm can be derandomized for the cost of additional multiplier of
  order $n$ in the running time.
\end{restatable}

\begin{restatable}{theorem}{thmglobalopt}
\label{thm:gen-alg}
If \textsc{Minimum $\Pi$} (\textsc{Maximum $\Pi$}, respectively) can be solved in time $T(|U|,\|\mathcal{F}\|,M)$ then \textsc{Minimum $\Pi$ with $k$ Free Cheap Elements} (\textsc{Maximum $\Pi$ with $k$ Free Expensive Elements}, respectively) can be solved in
time $\Oh(|U|\cdot T(|U|,\|\mathcal{F}\|,M))$ where $M$ is the maximum value of the cost function.
\end{restatable}

Hence, if \textsc{Minimum $\Pi$} is solvable in polynomial time, then
\textsc{Minimum $\Pi$ with $k$ Free Cheap Elements} is as well, and if
\textsc{Maximum $\Pi$} is, then so is
\textsc{Maximum $\Pi$ with $k$ Free Expensive Elements}.
By pipelining
\Cref{thm:gen-alg} with known complexity results for classic variants of
cut problems, we obtain a number of corollaries.

Since \minstcut{} is polynomial-time solvable in almost linear
time~\cite{ChenKLPGS23}, by \Cref{thm:gen-alg}, \minstcheapshort{} can be
solved in $\Oh(m^{2+o(1)}\log M)$ time for instances with cost functions
of maximum value $M$.  Furthermore, this result holds for directed
graphs.

Similarly, \mincut{} can be solved in polynomial time by the randomized
Karger's algorithm~\cite{karger1996newapproach} or the deterministic
Stoer--Wagner algorithm~\cite{stoer1997simple}.  In particular, using
the improved version of Karger's algorithm given by
\citet{GawrychowskiMW24}, we have that \mincheapshort{} is
solvable in $\Oh(m^2\log^2n\log M)$ time by a randomized algorithm for
instances with cost functions of maximum value $M$.

The classic \maxcut{} is \textsf{NP}-complete.  However, \maxcut{} can be
solved in polynomial time on planar graphs~\cite{hadlock1975finding},
graphs excluding $K_5$ as a minor~\cite{BARAHONA1983107}, and graphs of
bounded genus~\cite{galluccio2001optimization}.  Combining
\Cref{thm:gen-alg} and the results from \cite{LiersP12}, we obtain that \maxexpshort{}
can be solved in
$\Oh(n^{5/2}\log n\log M)$ time on planar graphs for instances with cost
functions of maximum value $M$.

\bigskip

By \Cref{thm:global_genus}, \minstexpshort{} is solvable in polynomial
time on graphs of bounded genus with polynomial costs.  We show that on
planar graphs, the constraint on costs can be excluded.  Interestingly,
already the first papers~\cite{wollmer1963methodsdetermining,mcmasters_optimal_1970} on
\minstexpshort studied planar graphs.  However, polynomial time
algorithms were known only for the situation when both terminal vertices
$s$ and $t$ lie on outer face.

\begin{restatable}{theorem}{thmplanarcutp}
  \label{thm:planar-vital}
  On planar graphs, \minstexpshort can be solved in time
  $\Oh(kn^2\log n \log{M})$ where $M$ is the maximum value of the cost
  function.
\end{restatable}

The core idea in the proof of \Cref{thm:planar-vital} is that computing
a discounted cut in a planar graph~$G$ can elegantly be reduced to finding
specific walks in its dual graph~$G^*$.  More specifically, every
minimal cut in~$G$ corresponds to a cycle in~$G^*$.  Let us fix an
$s$-$t$-path $P$ in~$G$.
The key observation (sometimes known as the \emph{discrete Jordan curve theorem})
is that a cycle in~$G^*$
corresponds to an $s$-$t$-cut in~$G$ if and only if it crosses~$P$ an
odd number of times.  Consequently, the problem reduces to finding a
minimum-cost closed walk in~$G^*$ that intersects~$P$ an odd number of
times, which can be solved using dynamic programming.

Finally, for \minexpshort, we give a polynomial-time algorithm that
works with arbitrary edge costs on general graphs.  In light of the
\textsf{NP}-hardness of \minstexp\ (see \Cref{thm:w-one-hard}), this
result highlights a surprising difference in the complexity of these two
problems.

\begin{restatable}{theorem}{thmglobalcut}
  \label{thm:global_mincut_p}
  \minexpshort admits a randomized algorithm that runs in time
  $\mathcal{O}(n^3\cdot m\cdot \log^4 n\cdot\log\log n\cdot \log M)$
  where $M$ is the maximum value of the cost function.
\end{restatable}

The proof of \Cref{thm:global_mincut_p} relies on an algorithm for the
\textsc{Bicriteria Global Minimum Cut}.  In this problem, we are given a
graph $G$ with two edge weight functions
$w_1,w_2\colon E(G)\rightarrow \mathbb{R}_{\ge 0}$ and real numbers
$b_1,b_2\geq 0$.  The task is to decide whether there is a cut $(A,B)$
of $G$ such that $w_i(E(A,B))\le b_i$ for every $i\in \{1,2\}$.  By 
 \cite{armon2006multicriteriaglobal},  this problem is
solvable in polynomial time.  (The algorithm of Armon and Zwick runs in
polynomial time even in a more general setting for an arbitrary fixed
number $k$ of criteria.  But for the proof of
\Cref{thm:global_mincut_p}, we need only $k=2$.)  The fastest known
algorithm for \textsc{Bicriteria Global Minimum Cut} is due to
\citet{AissiMR17}, and provides a solution in time
$ \Oh(n^3 \log^4{n}\log{\log {n}})$ with probability $1-1/\Omega(n)$.

\subsection{Related Work}

The problem of determining the \emph{most vital link}
(singular)~\cite{wollmer1963methodsdetermining} in a graph is the
following.  Given a flow network $(G, s, t, c)$, that is, a graph $G$
with source and target vertices, and a capacity function, determine
which edge in $G$ should be removed to obtain a graph with the smallest
possible maximum $s$-$t$-flow.  Clearly, this can be done in polynomial
time, since we can try every edge and run a max-flow algorithm.  By the
Max Flow--Min Cut theorem, we want to find an edge whose removal reduces
the Min $s$-$t$-cut as much as possible.  Wollmer studied the problem on
$s$-$t$-planar graphs, i.e., where the flow network is planar and $s$
and $t$ are both on the outer face.  Using a correspondence between
minimal cuts and paths in planar graphs, Wollmer provided a linear-time
algorithm for the problem on this class of graphs.

Wollmer extended the concept of most vital link to the \textsc{Most
  Vital Links} problem~\cite{wollmer1964removingarcs}, where the
objective is to remove $k$ edges to minimize the $s$-$t$-flow.  Again,
by the Max Flow--Min Cut theorem, this is equivalent to deleting $k$
edges to obtain a graph with as small a minimum $s$-$t$-cut as possible.
Therefore, we want to find a minimal cut and delete edges only from that
cut---$k$ edges with highest capacities.  The remaining edges will be
the resulting cut.  This implies that \textsc{Most Vital Links} is
equivalent to \minstexpshort.

\citet{phillips1993networkinhibition} studied this variant and
provided a Polynomial-Time Approximation Scheme (PTAS) for planar
graphs.  Later, 
 \citet{wood1993deterministicnetwork} claimed that
the \textsc{Most Vital Links} on planar graphs had been solved.
However, Wood's assertion merely referenced Wollmer's result, which
applies exclusively to $s$-$t$-planar graphs.  In fact, Wollmer's
algorithm assumes $s$ and $t$ share a face in the planar embedding, as
it starts by creating a direct edge between $s$ and $t$, which is not
feasible otherwise.

The \textsc{Most Vital Links} has been studied under various names
across different fields.  One such area is \emph{graph vulnerability},
which examines how much a graph ``breaks apart'' when a certain number
of edges or vertices are removed.  Another is \emph{network
  inhibition}~\cite{phillips1993networkinhibition} or \emph{network
  interdiction}~\cite{wood1993deterministicnetwork}, where the goal is
to ``damage'' a graph as effectively as possible while using minimal
resources.  
\citet{mcmasters_optimal_1970} study the
problem of allocating a limited number of aircraft to interdict an
enemy’s supply lines on a particular day using and give a linear
programming algorithm for the case where the cost of destroying an edge
is the same as its capacity.  Like Wollmer, they consider only
$s$-$t$-planar graphs.

Later, several variations of the \textsc{Most Vital Links} problem were
studied, mostly under the names \textsc{Network Interdiction} and
\textsc{Network Inhibition}.  Many of these problems were modeled as
Stackelberg games, in which an interdictor attempts to damage the
network while a defender tries to repair or reinforce
it~\cite{steinrauf1991networkinterdiction}.  \citet{zenklusen_network_2010} studied several
interdiction problems, including flow interdiction on planar
graphs and later, matching
interdiction~\cite{Zenklusen14}, i.e., removing edges to make the
maximal matching as small as possible.

Interestingly, despite the  \textsc{Most Vital Links} problem's long history and extensive
bibliography, several open algorithmic questions remain about the minimum
and maximum cuts whose cost is measured up to $k$ edges.  In this paper,
we address these gaps and introduce additional algorithmic questions,
hoping to stimulate further progress in this area.

\section{Preliminaries}
\label{sec:preliminaries}
For a positive integer $k$, we write $[k]$ to denote the set $\{1,
\ldots,k\}$.

\paragraph{Graphs.}
In this paper, we consider finite undirected
graphs and refer to the textbook by Diestel~\cite{Diestel12} for
undefined notions.  We use $V(G)$ and $E(G)$ to denote the set of
vertices and the set of edges of $G$, respectively.  We use~$n$ and~$m$
to denote the number of vertices and edges in~$G$, respectively, unless
this creates confusion, in which case we clarify the notation
explicitly.

For a vertex subset~$U \subseteq V$, we use~$G[U]$ to denote the
subgraph of~$G$ induced by the vertices in~$U$ and~$G - U$ to
denote~$G[V \setminus U]$.  For a vertex $v$, we use $\deg_G(v)$ to
denote its \emph{degree}, i.e, the number of its neighbors; we may omit
the index if $G$ is clear from the context.

A \emph{walk} $W=v_0\dots v_\ell$ of in a graph $G$ is a sequence of
(not necessarily distinct) vertices such that $v_{i-1}$ and $v_{i}$ are
adjacent for all $i\in[\ell]$.  We say that $W$ is \emph{odd} if $\ell$
is an odd integer.  We also say that $\ell$ is the \emph{length} of $W$
for unweighted graphs.  Given a cost/weight function
$c\colon E(G)\rightarrow\mathbb{Z}_{\geq 0}$, the length of
$W=v_0\dots v_\ell$ is defined as $\sum_{i=1}^\ell c(v_{i-1}v_i)$.  We
say that $W$ is \emph{closed} if $v_0=v_\ell$.  The vertices $v_0$
and~$v_\ell$ are the \emph{endpoints} of~$W$ and the other vertices are
\emph{internal}.  A \emph{path} is a walk with distinct vertices, and a
\emph{cycle} is a closed walk where all vertices except endpoint are
distinct.  Notice that a path $P=v_0\dots v_\ell$ is uniquely defined by
its set of edges $E(P)=\{v_{i-1}v_i\mid i\in[\ell]\}$, and the same
holds for a cycle.  Thus, we may say that a set of edges is a path
(cycle) if it is the set of edges of a path (a cycle, respectively).  We
also can consider a path or cycle to be a graph.

For two disjoint sets of vertices $A,B\subseteq V(G)$, we write
$E(A,B)=\{uv\in E(G)\mid u\in A,~v\in B\}$, An \emph{(edge) cut} of $G$
is a partition $(A,B)$ of $V(G)$ and the set of edges $E(A,B)$ is the
\emph{cut-set}.  We say that $(A,B)$ is an $s$-$t$-cut for two vertices
$s$ and $t$ if these vertices are in distinct sets of the partition.
Given a cost function $c\colon E(G)\rightarrow \mathbb{Z}_{\ge 0}$, the
\emph{cost} of a cut $(A,B)$ is defined as the sum of the cost of the
edges of $E(A,B)$.  A cut $(A,B)$ is an \emph{(inclusion) minimal} if
there is no other cut $(A',B')$ such that $E(A',B')\subset E(A,B)$.  It
is well-known that $(A,B)$ is minimal if and only if $G[A]$ and $G[B]$
are connected.

\paragraph{Planar graphs and graphs on surfaces.}
A graph is
\emph{planar} if it can be drawn on the plane such that its edges do not
cross each other.  Such a drawing is called a \emph{planar embedding} of
the graph and a planar graph with a planar embedding is called a
\emph{plane} graph.  The \emph{faces} of a plane graph are the open
regions of the plane bounded by a set of edges and that do not contain
any other vertices or edges.  We remind that a cycle of a plane graph
has exactly two faces.  For a plane graph~$G$, its \emph{dual}
graph~$G^*=(F(G),E^*)$ has the set of faces of $G$ as the vertex set,
and for each $e\in E(G)$, $G^*$ has the \emph{dual} edge~$e^*$ whose
endpoints are either two faces having~$e$ on their frontiers or~$e^*$ is
a self-loop at~$f$ if~$e$ is in the frontier of exactly one face~$f$
(i.e., $e$ is a bridge of $G$).  Observe that~$G^*$ is not necessarily
simple even if $G$ is a simple graph.  Throughout the paper, we will use
$\circ^*$ to denote the \emph{dual} of either a plane graph, a set of
vertices/faces, or a set of edges, or a single element.  Given a cost
function $c\colon E(G)\rightarrow \mathbb{Z}_{\ge 0}$, we assume that
$c(e^*)=c(e)$ for every $e\in E(G)$.

It is well known (see, e.g.,~\cite{Diestel12}) that finding an inclusion
minimal cut for a plane graph is equivalent to computing a cycle in the
dual graph.  More precisely, $(A,B)$ is a minimal cut of a connected
plane graph $G$ if and only if $\{e^*\in E(G^*)\mid e\in E(A,B)\}$ is a
cycle of $G^*$.  Furthermore, $(A,B)$ is a minimal $s$-$t$-cut if and
only if $C=\{e^*\in E(G^*)\mid e\in E(A,B)\}$ is a cycle of $G^*$ such
that $s$ and $t$ are in distinct faces of $C$.

The genus of a graph is the minimal integer $g$ such that the graph can
be drawn without crossing itself on a sphere with $g$ handles (i.e.\ an
oriented surface of the genus $g$).

\section{General Technique for Subset Problems with Free Elements}
\label{sec:general-techniques}

In this section, we present general results for minimization
(maximization, respectively) problems with free cheap (expensive,
respectively) elements.  We show that, given an algorithm $\mathcal{A}$
for \textsc{Minimum $\Pi$} (\textsc{Maximum $\Pi$}, respectively), we
can solve \textsc{Minimum $\Pi$ with $k$ Free Cheap Elements}
(\textsc{Maximum $\Pi$ with $k$ Free Expensive Elements}, respectively)
using at most $|U|$ calls of $\mathcal{A}$.

For an instance $(U,c,\mathcal{F})$ of \textsc{Minimum $\Pi$}
(\textsc{Maximum $\Pi$}, respectively), let
$\opt_{{\rm min}}(U,c,\mathcal{F})$
($\opt_{{\rm max}}(U,c,\mathcal{F})$, respectively) be the optimum cost
of a feasible solution, and we use
$\opt_{{\rm min}}^{{\rm cheap}}(U,c,\mathcal{F},k)$ and
$\opt_{{\rm max}}^{{\rm exp}}(U,c,\mathcal{F},k)$ for the optimum cost
for \textsc{Minimum $\Pi$ with $k$ Free Cheap Elements} and
\textsc{Maximum $\Pi$ with $k$ Free Expensive Elements}, respectively.
We assume that these values are zeros if optimum subsets are of
cardinality at most $k$.  Let $c\colon U\rightarrow\mathbb{Z}_{\geq 0}$
be a function and let $w\geq 0$ be an integer.  We use $c_w$ and $c^w$
for functions

\begin{minipage}{0.45\linewidth}\small
\[c_w(x)=
\begin{cases}
c(x)&\mbox{if }x> w\\
w&\mbox{if }x\leq w
\end{cases}
\]
\end{minipage}\hfill
\begin{minipage}{0.45\linewidth}\small
\[
c^w(x)=
\begin{cases}
c(x)&\mbox{if }x< w\\
w&\mbox{if }x\geq w
\end{cases}
\]
\end{minipage}

\smallskip\noindent
and write  $C(X)=\{c(x)\mid x\in X\}$ for $X\subseteq U$.

\newcommand{\Opt}{\ensuremath{{\mathrm{O}}}}

Now, let $$\Opt_{\min}{} = \min_{w\in C(U)}(\opt_{{\rm min}}(U,c_w,\mathcal{F})-kw)$$ and
$$\Opt_{\max}{} = \max_{w\in C(U)}(\opt_{{\rm max}}(U,c^w,\mathcal{F})-kw).$$

\begin{lemma}\label{lem:general}
For an instance $(U,c,\mathcal{F})$ of \textsc{Minimum $\Pi$} (\textsc{Maximum $\Pi$}, respectively) and $k\geq 0$,
\begin{itemize}
\item $\opt_{{\rm min}}^{{\rm cheap}}(U,c,\mathcal{F},k)=\max\{0, \Opt_{\min}{} \}$
\item $\opt_{{\rm max}}^{{\rm exp  }}(U,c,\mathcal{F},k)=\max\{0, \Opt_{\max}{} \}$.
\end{itemize}
\end{lemma}

\begin{proof}[Proof of \Cref{lem:general}]
  We prove the lemma for \textsc{Minimum $\Pi$} and \textsc{Minimum
    $\Pi$ with $k$ Free Cheap Elements}; the proof for the maximization
  variant uses symmetric arguments.

  First, we prove that
\begin{equation*}
\opt_{{\rm min}}^{{\rm cheap}}(U,c,\mathcal{F},k)\geq\max\{0,\min_{w\in C(U)}(\opt_{{\rm min}}(U,c_w,\mathcal{F})-kw)
\}.
\end{equation*}
Note that $\opt_{{\rm min}}^{{\rm cheap}}(U,c,\mathcal{F},k)\geq 0$ by
definition.  Thus, we have to show that
$\opt_{{\rm min}}^{{\rm cheap}}(U,c,\mathcal{F},k)\geq \min_{w\in
  C(U)}(\opt_{{\rm min}}(U,c_w,\mathcal{F})-kw)$.  Assume that
$S\subseteq U$ is a feasible set of optimum cost.  We consider two
cases.

Suppose that $|S|\leq k$.  Then
$\opt_{{\rm min}}^{{\rm cheap}}(U,c,\mathcal{F},k)=0$.  Consider
$w^*=\max C(U)$.  Then
$\opt_{{\rm min}}(U,c_{w^*},\mathcal{F})\leq kw^*$ because $S$ is a
feasible solution of size at most $k$.  Therefore,
$\opt_{{\rm min}}(U,c_{w^*},\mathcal{F})- kw^*\leq 0$ and
$\opt_{{\rm min}}^{{\rm cheap}}(U,c,\mathcal{F},k)\geq \opt_{{\rm
    min}}(U,c_{w^*},\mathcal{F})- kw^* \geq \min_{w\in C(U)}(\opt_{{\rm
    min}}(U,c_w,\mathcal{F})-kw)$.

Assume that $|S|>k$ and let $X\subseteq S$ be the subset of $k$ cheapest
elements of $S$.  We set $w^*=\max C(X)$.  Then
\begin{align*}
\opt_{{\rm min}}^{{\rm cheap}}(U,c,\mathcal{F},k)=&c(S\setminus X)=c_{w^*}(S\setminus X)\\= &c_{w*}(S)-c_{w^*}(X)=c_{w*}(S)-kw^*\\ &
\geq  \opt_{{\rm min}}(U,c_{w^*},\mathcal{F})-kw^*\\ &\geq
\min_{w\in C(U)}(\opt_{{\rm min}}(U,c_w,\mathcal{F})-kw).
\end{align*}

For the opposite inequality
\begin{equation*}
\opt_{{\rm min}}^{{\rm cheap}}(U,c,\mathcal{F},k)\leq\max\{0,\min_{w\in C(U)}(\opt_{{\rm min}}(U,c_w,\mathcal{F})-kw)
\},
\end{equation*}
assume first that there is a feasible set of size at most $k$.  Then
$$\opt_{{\rm min}}^{{\rm cheap}}(U,c,\mathcal{F},k)=0$$ by definition and
the inequality holds.  Suppose that any feasible set is of size at most
$k+1$.  In particular, this means that
$\max\{0,\min_{w\in c(U)}(\opt_{{\rm
    min}}(U,c_w,\mathcal{F})-kw)=\min_{w\in C(U)}(\opt_{{\rm
    min}}(U,c_w,\mathcal{F})-kw).  $ Let $w^*\in C(U)$ and let $S$ be a
feasible set such that
$\min_{w\in C(U)}(\opt_{{\rm min}}(U,c_w,\mathcal{F})-kw)= \opt_{{\rm
    min}}(U,c_{w^*},\mathcal{F})-kw^*=c_{w^*}(S)-kw^* $.  Denote by
$X\subseteq S$ the set of $k$ cheapest elements of $S$ with respect to
$c_{w^*}$.  By the definition of $c_{w^*}$, $c_{w^*}(x)\geq w^*$ for any
$x\in U$.  Then
$c_{w^*}(S)-kw^*\geq c_{w^*}(S)-c_{w^*}(X)=c_{w^*}(S\setminus X)$.
Since $c_{w^*}(x)\geq c_w(x)$ for all $x\in U$,
$c_{w^*}(S\setminus X)\geq c(S\setminus X)\geq \opt_{{\rm min}}^{{\rm
    cheap}}(U,c,\mathcal{F},k)$.  This proves the inequality and
completes the proof of the lemma.
\end{proof}

\Cref{lem:general} proves that \textsc{Minimum $\Pi$ with $k$ Free Cheap
  Elements} and \textsc{Maximum $\Pi$ with $k$ Free Expensive Elements}
can be reduced to \textsc{Minimum $\Pi$} and \textsc{Maximum $\Pi$},
respectively.  This proves \Cref{thm:gen-alg}.

\bigskip

In particular, if \textsc{Minimum $\Pi$} or \textsc{Maximum $\Pi$} can
be solved in polynomial time then the corresponding problem with free
edges admits a polynomial-time algorithm as well.  It is well-known that
\minstcut is polynomial-time solvable and, by the recent results of \citet{ChenKLPGS23}, the minimum $s$-$t$-cut can be found in an
almost linear time.  This implies the following proposition.

\begin{proposition}\label{prop:minstcut}
  \minstcheapshort can be solved in $\Oh(m^{2+o(1)}\log M)$ time for
  instances with cost functions of maximum value $M$.  Furthermore, this
  result holds for directed graphs.
\end{proposition}

Similarly, \mincut can be solved in polynomial time by the randomized
Karger's algorithm \cite{karger1996newapproach} or the deterministic
Stoer–Wagner algorithm~\cite{stoer1997simple}.  In particular, using the
improved version of Karger's algorithm~\cite{GawrychowskiMW24}, we obtain the following proposition.

\begin{proposition}\label{prop:mincut}
\mincheapshort can be solved in $\Oh(m^2\log^2n\log M)$
time by a randomized algorithm for instances with  cost functions of maximum value $M$.
\end{proposition}

Finally in this section, we note that \maxcut can be solved in
polynomial time on planar graphs~\cite{hadlock1975finding}, graphs
excluding $K_5$ as a minor~\cite{BARAHONA1983107}, and graphs of bounded
genus~\cite{galluccio2001optimization}.  We discuss this problem in
detail in the next section.  Here, we only mention that combining
\Cref{thm:gen-alg} and the results of \citet{LiersP12}, we obtain the following result for planar
graphs.

\begin{proposition}\label{prop:maxcut}
  \maxexpshort can be solved in $\Oh(n^{5/2}\log n\log M)$ time on
  planar graphs for instances with cost functions of maximum value $M$.
\end{proposition}

\subsection{Lower Bounds for \minstexp}
\label{sec:lower-bound-stexp}
In this section we prove that \minstcheap{} is \NP{}-complete and \W{}-hard even in very restricted cases.
Notice that the problem can be solved in
polynomial time when the edges of the input graph have unit costs---to
obtain an optimum solution, it is sufficient to find the minimum
$s$-$t$-cut in the considered graph and subtract~$k$ from its size.  We
prove that the problem becomes \textsf{NP}-hard when we allow edges to
have cost 1 or 2 and show that \minstexpshort is \textsf{W[1]}-hard when
parameterized by~$k$, the number of edges to delete, when we allow
polynomial (in the graph size) costs.

\thmhardness*

\begin{proof}
  First, we show \textsf{NP}-hardness and then explain how to modify the
  proof for the second claim.  We reduce from the \textsc{Clique}
  problem which asks, given a graph $G$ and a positive integer $k$,
  whether $G$ contains a clique of size $k$.  This problem is well-known
  to be \textsf{NP}-complete~\cite{GareyJ79} and \textsf{W[1]}-complete
  when parameterized by $k$~\cite{DowneyF13}.  Furthermore
  (see~\cite{MathiesonS12}), it is known that these computational lower
  bounds hold for regular graphs.  Notice that it is sufficient to prove
  \Cref{thm:w-one-hard} for multigraphs; to show the results for simple
  graphs, it is sufficient to subdivide each edge once and assign to the
  obtained two edges the same cost as the cost of the original edge.

  Let $(G,k)$ be an instance of the \textsc{Clique} problem where $G$ is
  a $d$-regular graph.  We set $p=(d-k+1)k$ and $q=2p+2$.  Then we
  construct the multigraph $G'$ with edge costs as follows.
\begin{itemize}
\item[(i)] Construct a copy of $G$, and assign cost one to each edge of
  $G$.
\item[(ii)] Add a vertex $s$, and for each $v\in V(G)$, add $q$ edges
  $sv$ of cost two.
\item[(iii)] Add a vertex $t$, and for each $v\in V(G)$, add $p+q+1$
  edges $tv$ of cost one.
\end{itemize}
We use $c$ to denote the cost function for the obtained instance of
\minstexpshort.  We define $k'=kq$ and set $\beta=p+(n-k)(p+q+1)$.  We
claim that $(G,k)$ is a yes-instance of \textsc{Clique} if and only if
there is an $s$-$t$-cut $(A,B)$ of $G'$ with
$\coste_{k'}(A,B)\leq \beta$.

Suppose that $G$ has a clique $K$ of size $k$.  We define
$A=\{s\}\cup (V(G')\setminus K)$ and $B=V(G')\setminus A$.  Since
$s\in A$ and $t\in B$, we have that $(A,B)$ is an $s$-$t$-cut.  The $k'$
most expensive edges in $E(A,B)$ are the edges $sv$ for $v\in K$; each
of them is of cost two.  The set $R=\{sv\mid v\in K\}$ contains $qk=k'$
edges.  Because $G'$ is $d$-regular,
$c(E(A,B))=qk+(d-k+1)k+(p+q+1)(n-k)=c(R)+p+(p+q+1)(n-1)=c(R)+\beta$.
Thus, $c(E(A,B))-c(R)=\beta$.  This means that
$\coste_{k'}(A,B)\leq \beta$.

For the opposite direction, assume that $\coste_{k'}(A,B)\leq \beta$ for
an $s$-$t$-cut $(A,B)$ of $G'$ assuming that $s\in A$ and $t\in B$.  Let
$X=V(G)\cap B$ and $Y=V(G)\cap B$.  We show that $X$ is a clique of $G$
of size $k$.

First, we show that $|X|=k$.  The proof is by contradiction.  Let
$r=|X|$ and consider two cases.

Suppose that $r>k$.  Then $k'$ most expensive edges are edges $sv$ for
$v\in X$.  The total number of edges of this type in $E(A,B)$ is
$qr>qk=k'$.  Let $R\subseteq E(A,B)$ be a set of $k'$ edges $sv$ for
$v\in X$.  Then $E(A,B)$ contains $(r-k)q$ edges
$sv\in E(A,B)\setminus R$ with $v\in X$ and the total cost of these
edges is $2(r-k)q$.  Also $E(A,B)$ contains $(n-r)(p+q+1)$ edges $tv$
for $v\in Y$ of cost one.  The total cost of these edges is
$(n-r)(p+q+1)$.  Hence,
\begin{align*}
c(E(A,B))-c(R) &\geq 2(r-k)q+(n-r)(p+q+1)\\
&= (r-k)(q+p+1)+(r-k)(q-p-1)\\
&\quad +(n-r)(p+q+1)\\
&= (r-k)(q-p-1)+(n-k)(p+q+1)\\
&\geq (q-p-1)+(n-k)(p+q+1)\\
&= p+1+(n-k)(p+q+1)>\beta
\end{align*}
contradicting that $\coste_{k'}(A,B)\leq \beta$.  Thus, $r\leq k$.

Suppose that $r<k$.  Then $k'$ most expensive edges are $rq$ edges $sv$
for $v\in X$ of cost two and $k'-rq=(k-r)q$ edges of cost one.  By the
construction of $G'$, $E(A,B)$ contains every edge $tv$ for $v\in Y$ and
the total cost of these edges is $(n-r)(p+q+1)$.  Therefore,
\begin{align*}
\coste_{k'}(A,B) &\geq (n-r)(p+q+1)-(k-r)q\\
&= (n-r)(p+q+1)-(k-r)(p+q+1)\\
&\quad +(k-r)(p+1)\\
&\geq (n-k)(p+q+1)+p+1>\beta;
\end{align*}
a contradiction.  We conclude that $|X|=k$.

Now we show that $X$ is a clique of $G$.  Because $|X|=k$, the set of
$k'$ the most expensive edges of $E(A,B)$ is $R=\{sv\mid v\in X\}$.
Also $E(A,B)$ contains $(n-k)(p+q+1)$ edges $tv$ for $v\in Y$.  Because
these edges are of cost one, this implies that
$c(E(A,B))-c(R)=|E(X,Y)|+(n-k)(p+q+1)\leq \beta=p+(n-k)(p+q+1)$.  Thus,
$|E(X,Y)|\leq p=(d-k+1)k$.  Because $G$ is a $d$-regular graph, we
obtain that $G[X]$ contains at least $\frac{1}{2}k(k-1)$ edges.  Since
$|X|=k$, we have that $X$ is a clique of size $k$.  This concludes the
\textsf{NP}-hardness proof for \minstexpshort.

To show that \minstexpshort is \textsf{W[1]}-hard when parameterized by $k$, we modify the above construction of $G'$ by replacing (ii) by (ii$^*$):
\begin{itemize}
\item[(ii$^*$)] Add a vertex $s$, and for each $v\in V(G)$, add  edge $sv$ of cost $q$.
\end{itemize}
Then we use the same arguments as in the \textsf{NP}-hardness proof to show that
$(G,k)$ is a yes-instance of \textsc{Clique} if and only if there is an $s$-$t$-cut $(A,B)$ of $G'$ with
$\coste_{k}(A,B)\leq \beta$.
This concludes the proof.
\end{proof}

\section{Graphs of Bounded Genus}\label{sec:genus}

In this section, we prove that all eight variants of the \emph{cuts with
  $k$ free edges} problems can be solved in polynomial time in graphs of
bounded genus, when restricted to polynomially-bounded integer edge
costs.  That is, our algorithms depend on $M$ polynomially, where $M$ is
the total edge cost.  We remark that our algorithms require an embedding
of $G$ into an orientable surface of genus $g$ given as an input; such
embeddings are known to be computable in time $2^{g^{\Oh(1)}} n$ \cite{mohar_linear_1999,KawarabayashiMR08}.

We rely heavily on the work of \citet{galluccio2001optimization}.  They show that integral total cut
costs can be computed efficiently---through polynomial evaluation and
interpolation---once the corresponding embedding is known.  We outline
their approach in the proof of the following lemma.

\begin{algorithm}
  \caption{Algorithm resolving an instance $(G,c,k,\beta)$ of a problem $\Pi$. If $\Pi$ is \textsc{$s$-$t$-Cut}, $s$ and $t$ are also a part of the input. Additionally, an embedding $\phi$ of $G$ into an oriented surface of genus $g$ is given.}
  \label{alg:genus-universal}
  \begin{algorithmic}[1]
    \STATE{$e_1,e_2,\ldots, e_m\leftarrow $ ordering of $E(G)$ s.t.\ $c(e_i) \leq c(e_{i+1})$}
    \STATE{$M\leftarrow 1+\sum_{i=1}^m c(e_i)$}
    \FOR{$t\in[m+1]$}
    \STATE{$G',\phi'\leftarrow$ copy of $G,\phi$}
    \FOR{$i\in[m]$}
    \IF{$i<t$ {\normalfont\textbf{and}} $\Pi$ is \textsc{Free Cheap} {\normalfont\textbf{or}}\\{\normalfont\textbf{if}} $i\ge t$ {\normalfont\textbf{and}} $\Pi$ is \textsc{Free Expensive}}
    \STATE{$c'(e_i)\leftarrow M$}
    \ELSE
    \STATE{  $c'(e_i)\leftarrow c(e_i)$}
    \ENDIF
    \IF{$\Pi$ is \textsc{$s$-$t$-Cut}}
    \STATE{add edge $e_{m+1}$ between $s$ and $t$ in $G'$}
    \STATE{attach a new handle between $s$ and $t$ in $\phi'$ and embed $e_{m+1}$ in it in $\phi'$}
    \STATE{$c'(e_{m+1})\leftarrow m\cdot M + 1$}
    \STATE{$T\leftarrow c'(e_{m+1})$}
    \ELSE
    \STATE{$T\leftarrow 0$}
    \ENDIF
    \ENDFOR
    \STATE{$C\leftarrow$ output of the algorithm of \Cref{lemma:genus_cut_interpolation} applied to $G'$, $c'$ and $\phi'$}
    \FOR{$w' \in \{w-T: w\in C, w\ge T\}$}
    \STATE{$\beta'\leftarrow w'\pmod  M$ \COMMENT{reimbursed cost}}
    \STATE{$k'\leftarrow \lfloor \frac{w'}{M}\rfloor$ \COMMENT{number of reimbursed edges}}
    \IF{$\beta'\le \beta$ {\normalfont \textbf{and}} $k'\le k$ {\normalfont \textbf{and}} $\Pi$ is \textsc{Min} {\normalfont \textbf{or}}\\ {\normalfont \textbf{if}} $\beta'\ge\beta$ {\normalfont \textbf{and}} $k'\ge k$ {\normalfont \textbf{and}} $\Pi$ is \textsc{Max}}
    \RETURN{\normalfont\textsc{Yes}}\label{line:encounter-cut}

    \ENDIF
    \ENDFOR
    \ENDFOR
    \IF{$\beta=0$ and $\Pi$ is \textsc{Max}}{\RETURN{\normalfont\textsc{Yes}}}
    \ENDIF
    \RETURN{\normalfont\textsc{No}}
  \end{algorithmic}
\end{algorithm}

\begin{lemma}\label{lemma:genus_cut_interpolation}
    There is a randomized algorithm that,
    given an $n$-vertex graph $G$ with positive integral edge costs, along with its embedding $\phi$ in an orientable surface of genus $g$,
in time
    $(4^g \cdot M \cdot n^{1.5}+M^2)\cdot \log^{\mathcal{O}(1)}(n+M),$
    where $M$ is the total edge cost, computes a set $C\subseteq \{0,1,\ldots, M\}$, that satisfies
    \begin{itemize}
        \item if $w\in C$, then there is a cut of cost $w$ in $G$;
        \item if there is a cut of cost $w$ in $G$, then $w\in C$ with probability at least $\frac{1}{2}$.
    \end{itemize}
    The algorithm can be derandomized for the cost of additional multiplier of order $n$ in the running time.
\end{lemma}

\begin{proof}
  Consider a polynomial
  $$\mathcal{C}(G,x)=\sum_{ \substack{A\cup B= V(G)\\ A\cap
      B=\emptyset}}x^{c(E(A,B))}=\sum_{w=0}^M a_wx^w.$$ Clearly, $G$ has
  a cut of cost $w$ if and only if the coefficient $a_w$ of $x^w$ in
  $\mathcal{C}(G,x)$ is non-zero.

  \citet{galluccio2001optimization}
  showed that for any prime $p$ and integer $i\in \{0,\ldots, p-1\}$,
  the value $\mathcal{C}(G,i)$ modulo $p$ can be computed in time
  $4^g\cdot n^{1.5}\cdot\log^{\mathcal{O}(1)}(n+M+p)$ if an embedding
  $\phi$ of $G$ into a surface of genus $g$ is given.

We now describe our algorithm.
First, it computes a sequence of pairwise-distinct prime numbers
\(p_{1},p_{2},\dots ,p_{2n}\), each greater than \(M\).
    This is done by iterating a number $q$ starting from $M+1$, adding $q$ to the sequence if it is prime, increase $q$ and repeat until the sequence has exactly $2n$ primes.
    Then the algorithm picks $j\in [2n]$ uniformly at random and puts $p:=p_j$.
Next, the algorithm selects an index \(j\in[2n]\) uniformly at random and sets
\(p:=p_{j}\).
For every \(i\in\{0,1,\dots ,M\}\), it calls the Galluccio--Loebl--Vondrák
subroutine to evaluate
$
  c_{p,i} = \mathcal{C}(G,i)\bmod p.
$
Using Lagrange interpolation with  arithmetic modulo~\(p\), it then   recovers the coefficients
\(a_{p,0},a_{p,1},\dots ,a_{p,M}\).
Finally, it outputs the set
$
  C_{p}\;:=\;\{\,w \mid a_{p,w}\neq 0\,\}.
$

    We argue that the algorithm is correct.
    Since $a_{p,w}\equiv a_w\pmod {p}$, $C_p$ can only contain numbers corresponding to non-zero coefficients in $\mathcal{C}(G,x)$ that correspond to existent cut costs.
    Consider $w\in \{0,\ldots, M\}$ such that $G$ admits a cut of total cost $w$.
    To see that $w\in C_p$ with probability at least $\frac{1}{2}$, note that $w\notin C_p$ if and only if $p$ divides $a_w$.
    The number of cuts is bounded by $2^n$, so $a_w\leq 2^n$, and at most $n$ primes can divide $a_w$, while the algorithm choice is made uniformly at random from a set of $2n$ distinct primes.

    We now explain the running time bound of the algorithm.
    The first part of the algorithm evaluates primes $p_1,p_2,\ldots, p_{2n}$.
    By properties of the prime number distribution, $p_{2n}$ is bounded by $\mathcal{O}((M+n)\cdot \log(M+n))$.
    Since a primality test can be performed in polylogarithmic time, the time required to construct the sequence is $\mathcal{O}(p_{2n}\cdot\log^{\mathcal{O}(1)}(p_{2n}))$.

    Next part of the algorithm evaluates $\mathcal{C}(G,x)$ in $M+1$ distinct values of $x$ modulo $p$.
    Each evaluation is a call to the Gallucio-Loebl-Vondr\'{a}k algorithm.
    The total running time of this part is bounded by $M\cdot 4^g\cdot n^{1.5}\cdot\log^{\mathcal{O}(1)}(n+M)$, since $\log p=\mathcal{O}(\log (M+n))$.

    Third subroutine of the algorithm corresponds to the construction of
    a Lagrangian polynomial of degree $M$ in $\operatorname{GF}(p)$.
    This is done in $\mathcal{O}(M^2)$ arithmetic operations modulo $p$.
    That is, in $M^2\cdot \log^{\mathcal{O}(1)}(M+n)$ running time.

    To see the upper bound from the lemma statement, sum up the upper bounds for second and third parts.
    Upper bound for the first part is dominated by the upper bound for the second part.

    To derandomize the algorithm, one just needs to run the algorithm $n+1$ times for every choice of $p$ among $p_1, p_2,\ldots, p_{n+1}$, and report the union of all obtained solutions $C_p$.
\end{proof}

The main result of this section is the following, restating from the introduction.

\thmglobalgenus*

\begin{proof}
    We show that each of the eight problems can be solved via $m+1$ calls to the algorithm of \Cref{lemma:genus_cut_interpolation}.

    Consider any ordering $e_1, e_2,\ldots, e_m$ of edges of $G$ such that $c(e_i)\le c(e_{i+1})$ for every $i\in[m-1]$.
    For a fixed threshold $t\in[m+1]$, we say that $e_i$ is \emph{cheap} if $i< t$ and  \emph{expensive} otherwise.
    For each of the eight problems, there is at least one correct choice of $t$, such that in an optimal solution cheap edges and expensive edges agree with the choice of $t$.
    That is, for \textsc{Free Cheap} problems, there is an optimal cut that contains at most $k$ edges $e_i$ with $i<t$, and reimbursed cost of this cut equals the sum of $c(e_i)$ for $i\ge t$ taken over edges that it contains.
    For \textsc{Free Expensive} problems, edges with $i\ge t$ are reimbursed, while costs of edges with $i<t$ are summed up in the final reimbursed cost.

Building on this idea, we present in \Cref{alg:genus-universal} a universal algorithm that depends only slightly on the specifics of the problem~$\Pi$.
Given an instance $(G,c,k,\beta)$ of $\Pi$ and embedding~$\phi$ of~$G$ into an orientable surface, the algorithm first finds an ordering of edges as discussed above.
    Then, for each choice of $t\in[m+1]$, the algorithms constructs a new graph~$G'$ with a new cost function~$c'$.
    Initially, $G'$ is a copy of $G$ and $\phi'$ is a copy of $\phi$, while $c'(e_i)=c(e_i)$ for edges that should not be reimbursed with respect to $t$ and $c'(e_i)=M$ for edges that should be reimbursed with respect to $t$.
    Depending on $\Pi$, these are either edges of form $e_i$ for $i<t$ or edges of form $e_i$ for $i\ge t$.
    Note that $M$ is slightly increased in a way that $M$ is strictly greater than the total edge cost.

    Then, if $\Pi$ is a \textsc{$s$-$t$-Cut} type of problem, the algorithm adds an extra edge $e_{m+1}$ between~$s$ and~$t$ in~$G'$.
    To extend $\phi'$ with $e_{m+1}$, the algorithm attaches a new handle between~$s$ and~$t$ and embeds $e_{m+1}$ in this handle in $\phi'$.
    This operation increases the surface genus by at most $1$, so $\phi'$ is an embedding of $G'$ in a surface of genus $g+1$.
    Then $e_{m+1}$ receives cost $c'(e_{m+1})=T$ for $T:=m\cdot M+1$.
    If $\Pi$ is a global cut type of problem, then $G'$, $\phi'$ remains the same, while the algorithm puts $T:=0$.

    As its next step, the algorithm runs the algorithm of \Cref{lemma:genus_cut_interpolation} as a subroutine applied to~$G'$ and~$c'$.
    A result of this run is a set $C$ of costs of some cuts of $(G',c')$.
    Then the algorithm looks for a total cost of form $T+k'\cdot M+\beta'$ in~$C$, for $ k'\ge 0$ and $
\beta'\le M$.
    If $k'\le k$ and $\beta'\le \beta$ (for $\Pi$ being a \textsc{Min} type of problem) or $k'\ge k$ and $\beta'\ge \beta$ (for $\Pi$ being a \textsc{Max} type of problem), the algorithm reports that the given instance is a yes-instance of $\Pi$ and stops.

    If the algorithm did not report a yes-instance for every choice of $t\in[m+1]$, then it finally reports that the given instance is a no-instance, with one exception: in case when $\beta=0$ and $\Pi$ is a \textsc{Max} type of problem, the algorithm reports a yes-instance.
    The description of the algorithm is finished.

    We claim that the algorithm recognizes yes-instances with high probability.

    \begin{claim}
        If $(G,c,k,\beta)$ is a yes-instance of $\Pi$, then the algorithm returns \textsc{Yes} with probability at least $\frac{1}{2}$.
    \end{claim}
    \begin{proof}
    First, we show the claim for $\Pi$ being a \textsc{Min} type and \textsc{Free Cheap} type of problem.
    To see this, let $(A,B)$ be a cut that certifies that the instance is a yes-instance.
    Let $S$ be the set of edges of this cut in $G$ and let $e_{j_1}, e_{j_2}, \ldots, e_{j_q}$ be the edges of $S$ in an order corresponding to the ordering found by the algorithm.
    Consider $t=j_{\min\{q,k\}}+1$.
    The reimbursed cost of $(A,B)$ equals $w=\sum_{e_i\in S, i\ge t}c(e_i)\le\beta<M$.
    Consider the cost of $(A,B)$ in $G'$ constructed for this choice of $t$.
    Let $S'$ be the set of edges of $(A,B)$ in $G'$.
    If $\Pi$ is \textsc{$s$-$t$-Cut}, then $e_{m+1}\in S'$, since $(A,B)$ is a cut between $s$ and $t$.
    Since, $S'\setminus \{e_{m+1}\}=S$, we have that the cost of $S'\setminus \{e_{m+1}\}$ is comprised of $\min\{q,k\}$ reimbursed edges that have cost $M$ in $(G',c')$, and other edges that have the same costs as in $(G,c)$ and constitute the reimbursed cost of $S$.
    Hence, the total cost of $(A,B)$ in $(G',c')$ equals $w'=T+\min\{q,k\}\cdot M+w$.
    With probability at least $\frac{1}{2}$, the algorithm of \Cref{alg:genus-universal} puts $w'$ in its resulting set $C$.
    Since $\min\{q,k\}\le k$ and $w\le \beta$, the algorithm correctly reports a yes-instance and terminates when it encounters $w'$ in $C$.

    When $\Pi$ is a \textsc{Max} type of problem, instances with $\beta=0$ are handled by the algorithm separately in its final step.
    We only have to consider instances with $\beta>0$.
    Then a solution cut $(A,B)$ always contains at least $k+1$ edges.
    Similarly to above, the cost of $S'$ in $(G',c')$ equals $w'=T+\min\{q,k\}\cdot M+w$, which now becomes $w'=T+k\cdot M+w$.
    The algorithm will encounter such $w'$ in $C$ with probability at least $\frac{1}{2}$ and report a yes-instance.

    For $\Pi$ being a \textsc{Free Expensive} type of problem the proof is symmetrical.
    \end{proof}

While we state \Cref{thm:global_genus} for all eight variants of discounted
cuts, some of these problems admit polynomial-time algorithms
on general graphs.  The first observation is that the variants of
``$\min -k$ cheap'' and ``$\max -k$ expensive'' are not significantly
different from their classic counterparts.  Furthermore, this type of
statement holds for optimization problems in the most general sense.

Consider the \textsc{Minimum $\Pi$} problem, whose task is, given a universe $U$, a cost function $c\colon U\rightarrow \mathbb{Z}_{\geq 0}$, and a family $\mathcal{F}$ of feasible subsets of $U$ (which may be given implicitly), to find the minimum cost of a feasible subset.
We define \textsc{Minimum $\Pi$ with $k$ Free Cheap Elements} as the problem whose input additionally includes an integer $k \geq 0$, and whose task is to find the minimum cost of a feasible subset, excluding the cost of the $k$ cheapest elements.
Similarly, we define \textsc{Maximum $\Pi$ with $k$ Free Expensive Elements} for the \textsc{Maximum $\Pi$} problem, where the objective is to find the maximum cost of a feasible subset while excluding the cost of the $k$ most expensive elements.

While \Cref{thm:global_genus} claims that eight problems admit polynomial-time algorithms, the algorithms are very similar and follow a common scheme.
In the proof below, we present a general algorithm that adjusts its constructions depending on a specific problem $\Pi$.

    Other than that, we have to show that the algorithm never returns \textsc{Yes} on no-instances.
    Assume that the algorithm returned \textsc{Yes} in line~\ref{line:encounter-cut} (the other yes-instance-return is trivially correct).
    By \Cref{lemma:genus_cut_interpolation}, there is a cut $(A,B)$ in $G'$ of total cost $w'=T+k'\cdot M+\beta'$.
    If $\Pi$ is \textsc{$s$-$t$-Cut}, then $(A,B)$ is an $s$-$t$-cut since $w'\ge T$.
    Also, $(A,B)$ contains exactly $k'$ edges of form $e_i$ for $i<t$ or $i\ge t$ depending on the reimbursement type of $\Pi$; the total cost of other edges equals $\beta'$.
    If $\Pi$ is \textsc{Min} type, then the total cost of $(A,B)$ without $k'\le k$ reimbursed edges in $(G,c)$ is exactly $\beta'$, hence with $k$ reimbursed edges the cost of $(A,B)$ can't become greater than $\beta'\le \beta$.
    Symmetrically, if $\Pi$ is \textsc{Max} type, then the total cost of $(A,B)$ is $\beta'$, if $k'\ge k$ edges are reimbursed, an with $k$ edges it can't become smaller.
    In either case, $(A,B)$ certifies that $(G,c,k,\beta)$ is a yes-instance of $\Pi$.
    The correctness of the algorithm follows.

    The running time required to find the ordering and construct $G'$, $c'$ and $\phi'$ is dominated by the running time of the algorithm of \Cref{lemma:genus_cut_interpolation}.
    Since the maximum cost in $c'$ is now bounded by $\Oh(m\cdot M)$, $G'$ is embedded in a surface of genus $g+1$, one call to the subroutine takes $$4^{g+1}\cdot n^{1.5}\cdot m \cdot M+m^2\cdot M^2$$
    running time up to a factor polylogarithmic in $n+M$.
    Our algorithm makes at most $m+1$ calls to the subroutine before it terminates.
    Hence, its total running time is upper-bounded by
    $$ \left(4^{g}\cdot n^{1.5} +m\cdot M\right)\cdot m^2\cdot M\cdot\log^{\Oh(1)}(n+M).$$

    Note that the derandomization of the algorithm of \Cref{lemma:genus_cut_interpolation} also derandomizes our algorithm for the cost of an additional multiplier of order $n$.
    This concludes the proof of \Cref{thm:global_genus}.
  \end{proof}

While \Cref{thm:global_genus} asserts that the eight problems admit polynomial-time algorithms, these algorithms share similarities and adhere to a common framework.

\newcommand{\Zpos}{\ensuremath{\mathbb{Z}_{> 0}}}

\subsection{Planar \minstexp}
\label{sec:planar-vital-links}
Now we prove \Cref{thm:planar-vital} which shows that \minstexpshort is
solvable in polynomial time on planar graphs with exponential edge costs.
Throughout the section, we assume that all graphs considered are plane, meaning
they come equipped with an embedding as the planarity can be tested and a planar embedding can be found (if it exists) in linear time, as shown by \citet{HopcroftT74}.

The crucial idea behind our algorithm is to switch to the dual problem.  Recall
that~$(A,B)$ is a minimal $s$-$t$-cut in a connected graph~$G$ if and only if
$C=\{e^*\in E(G^*)\mid e\in E(A,B)\}$ is a cycle of~$G^*$ such that~$s$ and~$t$
are in distinct faces of~$C$.
Observe that given an instance $(G,c,s,t,k,\beta)$
of \minstexpshort, an $s$-$t$-cut $(A,B)$ of minimum discounted cost $\coste_k(A,B)$
should be minimal.  Thus, the problem reduces to finding a cycle~$C$ in~$G^*$
of minimum discounted cost such that~$s$ and~$t$ are in distinct faces of~$C$.
Our first task is to establish a criteria that guarantees that~$s$ and~$t$ are
in distinct faces of a cycle of~$G^*$.

Let~$G$ be a plane graph.  We say that a path~$P$ in~$G$ \emph{crosses} a
cycle~$C^*$ of~$G^*$ in~$e\in E(P)$ if~$C^*$ contains the edge $e^*\in E(G^*)$
that is dual to~$e$.  The \emph{number of crosses} of~$P$ and~$C^*$ is the
number of edges of~$P$ where~$P$ and~$C^*$ cross.
We can make  the following observation.

\begin{observation}[\citet{BentertDFGK24}]\label{obs:cycle-sep}
  Let~$G$ be a plane graph, let~$s,t\in V(G)$, and let~$P$ be an $s$-$t$-path.
  For any cycle~$C^*$ of~$G^*$,~$s$ and~$t$ are in distinct faces of~$C^*$ if
  and only if the number of crosses of~$P$ and~$C^*$ is odd.
\end{observation}

\noindent
Thus, to solve \minstexpshort for $(G,c,s,t,k,\beta)$, we have to find a cycle~$C$
in~$G^*$ of minimum discounted cost with odd number of crosses with some
$s$-$t$-path in~$G$.
To solve this problem, we use the polynomial Dijkstra-style algorithm finding
an odd $s$-$t$-walk of minimum discounted cost in a general graph.  Let~$G$ be
a graph with a given cost function $c\colon E(G)\rightarrow \mathbb{Z}_{\geq 0}$.  Let also
$k\geq 0$ be an integer.  For a walk $W=v_0 \dots\ v_\ell$, let~$R$ be the
sequence of~$k$ most expensive edges $v_{i-1}v_i$ in~$W$; note that~$R$ may
include several copies of the same edge of~$G$.
Then we define $\coste_k(W)=\sum_{i=1}^\ell c(v_{i-1}v_i)-\sum_{e\in R}c(e)$.

We note that the following lemma holds for all graphs:

\begin{lemma}
  \label{lem:odd-walk-algorithm}
  Given a graph~$G$ with a cost function $c\colon E(G)\rightarrow \mathbb{Z}_{\geq 0}$, two
  vertices~$s$ and~$t$, and an integer $k\geq 0$, an odd $s$-$t$-walk~$W$ with
  minimum $\coste_k(W)$ can be found in $\Oh(mk \log n\log M)$ time where~$M$ is
  the maximum value of the cost function.
\end{lemma}

\begin{proof}[Proof \Cref{lem:odd-walk-algorithm}]

  We present a Dijkstra-style dynamic programming algorithm
  (\Cref{alg:odd-walk}).  The DP table $\dist(k', p, v) = \ell$ records the
  minimum length~$\ell$ of a walk from~$s$ to~$v$ that uses~$k'$ free edges and
  has parity $p \in \{0,1\}$.

  \medskip
  \noindent
We work with \emph{vectors} of the form
  $(\ell, k', p, v)$, which denotes that the shortest path from~$s$ to~$v$
  using a number of edges congruent to~$p$ where~$k'$ edges are free, has cost~$\ell$.

  \begin{claim}
    $\dist(k', p, v) = \ell$ is correct when the algorithm terminates.  In
    addition, the first time $(\circ, k', p, v)$ is popped, its value is
    correct.
  \end{claim}
  \begin{proof}[Proof of claim]
    Notice that $(0,0,0,s)$ is the first vector to be popped, and is correct.  Suppose that $(\ell, k, p, v)$ is the first time any $(\circ, k, p, v)$ vector is popped and is incorrect, and furthermore, is the first such incorrect vector to be popped.

    Then there either (a) exists a walk from~$s$ to~$v$ using~$k$ free
    edges with parity~$p$ with length less than~$\ell$, or (b) there is
    no such $s$-$v$-walk with length~$\ell$.

    In either case, $(\ell, k, p, v)$ was necessarily pushed onto the queue by
    a neighbor $u \in N(v)$ whose active vector at the time was $(\ell', k' \overline p, u)$, which has either (i)
    $(\ell - w(u,v), k, \overline p, u)$, or (ii) $(\ell, k-1, \overline p, u)$.
Since  $(\ell', k', \overline p, u)$ was popped prior to $(\ell, k, p, v)$, it is correct by assumption, hence
    $(\ell, k, p, v)$ is correct as well, contradicting the assumption.

    It follows that the first (and only) time we write $(k, p, v)$ to the
    $\dist$ map, the value is correctly written.
  \end{proof}

  We now analyze the running time of \Cref{alg:odd-walk}, aiming to show it is bounded by $\Oh(mk \log n \log M)$.
If we can bound the number of iterations
  the of the while loop to~$t$, we are done, since the work done inside the
  loop then is bounded by $\Oh(\deg(v) \cdot \log t)$ where~$v$ is the vertex
  popped on Line~5.

  \begin{claim}
    $v$ is pushed at most $2k\deg(v)$ many times.
  \end{claim}
  \begin{proof}
    When a neighbor $u \in N(v)$ pushes~$v$ onto the queue,~$u$ is popped off,
    and by the previous claim,~$u$s vector is correct for its current~$k$
    value, and hence will not push~$v$ more than that time.  When it pushes~$v$
    onto the queue, it pushes it at most twice, one with~$k'=k$ and one with
    $k'=k+1$, hence~$v$ is pushed at most~$2k\deg(v)$ many times.
  \end{proof}

  From the claim above, each vertex is popped at most $4k\deg(v)$ many times,
  hence each vertex is pushed at most $4k\deg(v)$ times.  In total
  $8k m = \Oh(m k)$ many push and pop operations are performed, each running in
  time $\Oh(\log mk) = \Oh(\log n)$.  Since the remaining computations are all
  $\Oh(\log M)$, the total running time is $\Oh(mk \log n \log M)$.
\end{proof}

We also use the following folklore observation.

\begin{observation}\label{prop:simple-odd-cycle}If $G=(V,E)$ contains an odd closed walk $W=v_0\dots v_\ell$, then it also
  contains an odd cycle~$C$ that is part of~$W$, that is, $V(C)\subseteq \{v_0,\ldots,v_\ell\}$
  and for every $e\in E(C)$, there is $i\in[\ell]$ such that $e=v_{i-1}v_i$.
\end{observation}
\begin{proof}
  Let $W$ be a closed odd walk and consider the graph~$B$ with
  $V(B)=\{v_0,\ldots,v_\ell\}$ and
  $E(B)=\{e\in E(G)\mid e=v_{i-1}v_i\text{ for some }i\in[\ell]\}$, i.e., the
  graph containing the vertices and all edges in the walk.  Since~$W$ is a
  closed odd walk,~$B$ is not bipartite (since there are no closed odd walks in
  a bipartite graph), hence~$B$ has an odd cycle~$C$.  This completes the
  proof.
\end{proof}

\begin{algorithm}
\caption{Cheapest odd walk with $k$ free edges}\label{alg:odd-walk}
\begin{algorithmic}[1]
\STATE{ Create a priority queue $\PQ$ }
\STATE{ $\PQ.\text{push}(0,0,0,s)$ }
\STATE{ $\dist(v,k,p) \gets \infty$ \hfill \COMMENT{all values initialized to $\infty$} }

\WHILE{$\PQ \neq \emptyset$}
  \STATE{ $(\ell, k', p, v) \gets \PQ.\text{pop}()$ }
  \IF{$\dist(k', p, v) \le \ell$}
    \STATE{ \textbf{continue} \hfill \COMMENT{discard sub-optimal vector} }
  \ENDIF
  \STATE{ $\dist(k', p, v) \gets \ell$ }
  \FOR{$u \in N_G(v)$}
    \STATE{ $\PQ.\text{push}(\ell + w(u,v), k', \overline{p}, u)$ }
    \IF{$k' < k$}
      \STATE{ $\PQ.\text{push}(\ell, k'+1, \overline{p}, u)$ }
    \ENDIF
  \ENDFOR
\ENDWHILE
\end{algorithmic}
\end{algorithm}

Now we are ready to prove \Cref{thm:planar-vital} which we restate:

\thmplanarcutp*

\begin{proof}
  Let $G$ be a plane graph with a cost function
  $c\colon E(G)\rightarrow \mathbb{Z}_{\geq 0}$, and let~$s$ and~$t$ be two distinct
  vertices, and $k\geq 0$ be an integer.  We assume that~$G$ is connected as
  \minstexp is trivial if~$s$ and~$t$ are in distinct connected components.

  We find an arbitrary $s$-$t$-path~$P$ in~$G$.  Then we construct the dual
  graph~$G^*$ with the corresponding cost function~$c^*$.  For each edge
  $e^*=uv\in E(G^*)$ such that the dual edge $e\notin E(P)$, we subdivide this
  edge, that is, we introduce a new vertex~$w$ and replace~$uv$ with~$uw$ and
  $wv$.  Denote the obtained graph~$G^\star$.  We define the cost function
  $c^\star\colon E(G^\star)\rightarrow\mathbb{Z}_{\geq 0}$ by setting
  $c^\star(e^*)=c^*(e^*)=c(e)$ for $e\in E(P)$, and we set
  $c^\star(uw)=c^*(uv)$ and $c^\star(wv)=0$ if~$uw$ and~$wv$ were obtained by
  subdividing~$e^*=uv$.  The construction of~$G^\star$ and the definition of
  the cost function~$c^\star$ immediately imply the following property.

\begin{claim}
  The graph~$G^*$ has a cycle~$C^*$ with odd number of crossing of~$P$ if and
  only if~$G^\star$ has an odd cycle~$C^\star$.  Moreover, $\min\coste_k(C^*)$
  taken over all cycles~$C^*$ with odd number of crossing of~$P$ is the same as
  $\min\coste_k(C^\star)$ where the minimum is taken over all odd
  cycles~$C^\star$ in~$G^\star$.
\end{claim}

Thus, we reduced our initial task to finding an odd cycle in~$G^\star$ of
minimum discounted cost.  To find such a cycle, we use the algorithm
from~\Cref{lem:odd-walk-algorithm}.  We go trough every vertex $x\in V(G^*)$,
and run the algorithm to find an odd $x$-$x$-walk~$W$ with minimum
$\coste_k(W)$.  Among all such walks, we find a closed walk~$W^\star$ of
minimum discounted cost.
By \Cref{prop:simple-odd-cycle}, there is an odd
cycle $C^\star$ that is a part of~$W^\star$.
Notice that
$\coste_k(C^\star)\leq \coste_k(W^\star)$.  Because~$W^\star$ is an odd closed
walk of minimum discounted cost, we obtain that~$C^\star$ is an odd cycle in
$G^\star$ of minimum discounted cost $\coste_k(C^\star)=\coste_k(W^\star)$.
Since the minimum value of
$\coste_k(A,B)$ of an $s$-$t$-cut $(A,B)$ of~$G$ is the same as
$\coste_k(C^\star)$, we return the cut
corresponding to the cycle~$C^\star$ in~$G^\star$.  This concludes the
description of the algorithm and its correctness proof.
It is easy to verify that the running is $\Oh(kn^2\log n\log M)$.

\end{proof}

\section{Global \minexp}
\label{sec:globally-most-vital}

This section is devoted to the \emph{global} version of the minimum cut problem with reimbursed edges, namely \minexpshort.
The problem is well-known for the case $k=0$, and the state-of-the-art randomized algorithm is due to \citet{GawrychowskiMW24} and runs in  $\mathcal{O}(m \cdot \log^2 n\cdot \log M)$ time, where $M$ is the maximum edge cost value.

The \emph{multi-criteria} global minimum cut problem, where edges have
multiple weight types that are measured separately, was introduced by
\citet{armon2006multicriteriaglobal} and later improved
by \citet{AissiMR17}.  In this work, we
focus on the two-criteria case and observe that \minexpshort is a
special case of the \textsc{Bicriteria Global Minimum Cut} problem:
Given a graph with two edge weight functions $w_1$ and $w_2$, does there exist a cut $(A,B)$
  of $G$ such that $w_i(E(A,B))\le b_i$ for every $i\in \{1,2\}$?

\begin{proposition}[\citet{AissiMR17}]\label{prop:bicriteria_poly}
    \textsc{Bicriteria Global Minimum Cut} admits a randomized algorithm with $\mathcal{O}(n^3 \cdot \log^4 n\cdot\log\log n\cdot \log M)$ running time, where $M$ is the maximum edge cost value.
\end{proposition}

We use \Cref{prop:bicriteria_poly} to prove
\Cref{thm:global_mincut_p}.
The basic idea is to consider $m+1$ choices of splitting edges into \emph{cheap} and \emph{expensive} and construct an instance of \textsc{Bicriteria Global Minimum Cut} for each choice.
Before giving the proof, we restate \Cref{thm:global_mincut_p} for the reader's convenience.

\thmglobalcut*

\begin{proof}
    We present an algorithm $\mathcal{A}$ that makes $m+1$ calls to the algorithm of \Cref{prop:bicriteria_poly}.
    Denote the given instance of \minexpshort by $(G,c,k,W)$.
    First, the algorithm finds an ordering $e_1, e_2, \ldots, e_m$ of edges of $G$ such that their costs are non-decreasing: $c(e_1)\le c(e_2)\le\ldots\le c(e_m)$.
    This is done via any sorting algorithm that runs in $\mathcal{O}(m\log m)$ time.

    Then, the algorithm iterates $t$ over values in $\{1,2,\ldots, m,m+1\}$.
    For a fixed choice of $t$, the algorithm constructs an instance $\mathcal{I}_t$ of \textsc{Bicriteria Global Minimum Cut}.
    The instance $\mathcal{I}_t$ is given by $(G,w_1^t, w_2^t, W, k)$.
    The idea is that the first weight function $w_1^t$ corresponds to the original edge costs of cheap edges given by $c$, while the second weight function $w_2^t$ corresponds to the number of reimbursed expensive edges.
    For every $i<t$, the algorithm defines $w_1^t(e_i):=c(e_i)$, while in the same time $w_2^t(e_i):=0$.
    These edges correspond to the edges that are included in the cost.
    For every $i\ge t$, we define $w_2^t(e_i):=1$, while $w_1^t(e_i)$ receives value zero.
    These edges correspond to the reimbursed edges.

    The algorithm $\mathcal{A}$ passes each of the instances $\mathcal{I}_1, \mathcal{I}_2, \ldots, \mathcal{I}_{m+1}$ to the algorithm of \Cref{prop:bicriteria_poly}.
    If at least one of these instances is a yes-instance, $\mathcal{A}$ reports that $(G,c,k,W)$ is a yes-instance.
    Otherwise, $\mathcal{A}$ reports that it is a no-instance.

    The running time upper bound for $\mathcal{A}$ is clear.
    We prove that $\mathcal{A}$ is correct.
    First, suppose $(G,c,k,W)$ is a yes instance of \minexpshort.
    Let $(A,B)$ be the solution to  $(G,c,k,W)$.
    Let $F$ be the subset of expensive edges of $E(A,B)$ that are free, while $P$ is the set of edges that are included in the cost.
    Since $c(e_i) \le c(e_j)$ for each $e_i \in P$, $e_j\in F$, we can assume $i<j$, i.e.\ the edge $e_i$ is always reimbursed \emph{after} the edge $e_j$ for $i<j$.
    We have $|F|\le k$ and $c(P)\le W$ (note that $|F|<k$ only if $|P|=0$).
    Let $t$ be the smallest number in $\{i: e_i\in F\}$, or $m+1$ if $F$ is empty.
    Then the same cut $(A,B)$ is a solution to $\mathcal{I}_t$: \[w_1^t(P)+w_1^t(F)=w_1^t(P)+0=c(P)\le W,\]
    and
    \[w_2^t(P)+w_2^t(F)=0+w_2^t(F)=|F|\le k.\]

    Towards opposite direction, let $\mathcal{I}_t$ be a yes-instance of \textsc{Bicriteria Global Minimum Cut} for some $t \in \{1,\ldots, m+1\}$.
    Let $(A,B)$ be its solution cut.
    Split $E(A,B)$ in $F$ and $P$, $F$ consists of all edges $e\in E(A,B)$ with $w^t_2(e)=1$, and $P$ consists of all other edges in $E(A,B)$.
    By definition of $\mathcal{I}_t$, $c(e)\le c(e')$ for every $e\in P$, $e'\in F$.
    In the same time, $c(P)=w_1^t(E(A,B))\le W$ and $|F|\le k$.
    Consequently, $\coste_{k}(A,B)\leq W$ and $(G,c,k,W)$ is a yes-instance of \minexpshort.
    The proof is complete.
\end{proof}

\section{Conclusion}\label{section:conclustion}

We initiated a comprehensive study of cut problems with discounts and conclude with three open problems.
We showed that on general graphs, \minstexpshort is  \textsf{W[1]}-hard parameterized by $k$  for instances with integer edge costs upper bounded by a polynomial of the size of the input graph. Does the problem become fixed-parameter tractable \textsf{FPT} when the edge costs are constants? Is  the problem \textsf{FPT} when parameterized by combined parameters $k$ and the maximum edge cost?
Our polynomial time algorithm for \maxcheapshort on graphs of bounded genus relies on the assumption that the costs are polynomial. Is the problem polynomial time solvable on planar graphs for arbitrary weights?

\end{document}